\documentclass[11pt]{article}

\usepackage{fullpage}

\usepackage[dvipsnames,svgnames,table]{xcolor}
\usepackage[colorlinks=true,linkcolor=RawSienna,citecolor=ForestGreen]{hyperref}

\usepackage[utf8]{inputenc}

\usepackage{amsthm}
\usepackage{thmtools}
\usepackage[normalem]{ulem}
\usepackage{mathtools,bm}
\usepackage[inline]{enumitem}
\usepackage{amsmath,amsfonts,amssymb,amsthm,hyperref,algorithm,eqparbox,array,ragged2e,stmaryrd,bbm}
\usepackage{xspace}
\usepackage{xcolor}
\usepackage{dsfont}
\usepackage{booktabs}
\usepackage[labelfont=bf]{caption}

\usepackage{tocloft}
\usepackage[nottoc]{tocbibind} 

\usepackage{enumitem}
\usepackage[capitalise,nameinlink,noabbrev]{cleveref}

\usepackage{tikz}

\declaretheorem[name=Theorem]{theorem}
\declaretheorem[sibling=theorem,name=Lemma]{lemma}

\declaretheorem[sibling=theorem,name=Corollary]{corollary}

\declaretheorem[sibling=theorem,name=Fact]{fact}

\declaretheorem[sibling=theorem,name=Claim]{claim}

\theoremstyle{definition}

\declaretheorem[sibling=theorem,name=Open Problem]{problem}

\theoremstyle{remark}

\makeatletter
\def\th@example{%
  \thm@notefont{}
  \normalfont 
}
\def\th@definition{%
  \thm@notefont{}
  \normalfont 
}
\makeatother

\newcommand{\fstr}{f_{\mathrm{str}}}
\newcommand{\fpsd}{f_{\mathrm{psd}}}

\renewcommand{\Pr}[1]{\mathbb{P} \left[ #1 \right]} 

\DeclareMathOperator{\rank}{rank}

\DeclareMathOperator{\disc}{disc}
\DeclareMathOperator{\Disc}{Disc}

\DeclareMathOperator{\Rcc}{R}
\DeclareMathOperator{\Dcc}{D}
\newcommand{\signrank}{\rank_\pm}
\newcommand{\suprank}{\rank_0}

\DeclareSymbolFont{yhlargesymbols}{OMX}{yhex}{m}{n} \DeclareMathAccent{\yhwidehat}{\mathord}{yhlargesymbols}{"62}
\newcommand{\TV}{\mathrm{dist}_{\mathrm{TV}}}

\newcommand{\newclass}[2]{\newcommand{#1}{{\text{\upshape\sffamily #2}}\xspace}}
\newclass{\Pcc}{P}
\newclass{\BPP}{BPP}
\newclass{\UPP}{UPP}
\newclass{\SUPP}{SUPP}

\newcommand{\bR}{\mathbb{R}}

\newcommand{\Oh}{{O}} 

\newcommand{\N}{\mathbb{N}}
\newcommand{\Z}{\mathbb{Z}}
\newcommand{\R}{\mathbb{R}}
\newcommand{\C}{\mathbb{C}}

\newcommand{\T}{\mathbb{T}}
\newcommand{\One}{\mathds{1}}
\newcommand{\cP}{\mathcal{P}}

\newcommand{\cD}{\mathcal{D}}
\newcommand{\cX}{\mathcal{X}}
\newcommand{\cY}{\mathcal{Y}}
\newcommand{\cZ}{\mathcal{Z}}

\newcommand{\define}[1]{{\itshape #1}}

\renewcommand{\epsilon}{\varepsilon}
\renewcommand{\leq}{\leqslant}
\renewcommand{\geq}{\geqslant}
\renewcommand{\setminus}{\backslash}

\renewcommand{\P}{\mathbb{P}}
\renewcommand{\subset}{\subseteq}

\newcommand{\E}{\mathbb{E}}

\usepackage[normalem]{ulem}

\usepackage[makeroom]{cancel}

\newcommand{\ZN}{\Z_N}
\newcommand{\PL}{\textsc{PL}}
\newcommand{\IIP}{\textsc{IIP}}
\renewcommand{\And}{\textsc{And}}

\begin{document}
\title{\mbox{}\\[0em]
\huge No Constant-Cost Protocol for Point--Line Incidence}

\author{
\mbox{}\\[-0.5em]
\setlength{\tabcolsep}{10pt}
\begin{tabular}{cccc}
Mika G\"o\"os&
Nathaniel Harms&
Florian~K.~Richter&
Anastasia Sofronova\\[-1mm]
\small\slshape EPFL &
\small\slshape UBC &
\small\slshape EPFL &
\small\slshape EPFL
\end{tabular} 
}

\date{
\mbox{}\\[0em]
\today}

\maketitle

\begin{abstract}
\noindent
Alice and Bob are given $n$-bit integer pairs $(x,y)$ and $(a,b)$, respectively,
and they must decide if $y=ax+b$. We prove that the randomised communication complexity
of this \emph{Point--Line Incidence} problem is $\Theta(\log n)$.
This confirms a conjecture of Cheung, Hatami, Hosseini, and Shirley (CCC~2023)
that the complexity is super-constant,
and gives the first example of a communication problem with constant
support-rank but super-constant randomised complexity.
\end{abstract}

\vspace{3em}

\setlength{\cftbeforesecskip}{0pt}
\renewcommand{\cftsecpagefont}{\normalfont}
\setcounter{tocdepth}{2}
\tableofcontents

\thispagestyle{empty}
\setcounter{page}{0}
\newpage

\section{Introduction}

Given $n$-bit integers $x,y,a,b$, how hard is it to check whether
\vspace{-0.4em}
\[
    y = ax + b ?
\]
In communication complexity, we suppose Alice has $(x,y)$ and Bob has $(a,b)$.
How many bits of communication are required for them to check this equation?
This is the simplest arithmetic (in)equality whose randomised
communication complexity is not yet known. For comparison, the randomised complexities of $y = ax$, $y \geq b$, and $y \geq ax+b$ are
fully understood; see~\cref{tab}.

The problem of deciding $y = ax+b$ is known as \textsc{Point--Line Incidence}
($\PL$).
To our annoyance, it is \emph{a priori} unclear if this can be solved by
even a \emph{constant-cost} randomised protocol, that is, with cost independent
of $n$. Indeed, there is no known structural criterion that rules out a constant-cost protocol for $\PL$, despite it being such a natural problem (see~\cref{sec:constant-cost} for the ongoing quest to find such structural criteria).
Cheung, Hatami, Hosseini, and Shirley \cite{Cheung23} conjectured that it does
not have a constant-cost protocol. We confirm their conjecture:

\begin{theorem}
\label{thm:cc-main}
The public-coin randomised communication complexity of $\PL$ is $\Omega(\log n)$.
\end{theorem}

This is tight: there is an $O(\log n)$-bit randomised protocol as observed
by~\cite{CLV19}. The players choose a random prime $p\leq
O(n)$, Alice sends $x$ and $y$ modulo~$p$, and Bob checks whether $y = ax+b
\pmod{p}$. This protocol has one-sided error, because if $y = ax+b$ then this
remains true modulo $p$.

The same idea gives a protocol for the more general
problem of checking $\sum_{i=1}^k x_iy_i = 0$ for $n$-bit integers $x_i, y_i$.
This is \textsc{Integer Inner Product} \cite{CLV19,Cheung23,Cheung25ITCS}.
As a corollary, we get tight bounds on this problem as well, stated in \cref{sec:intro-iip}.
These results witness the first separation between randomised communication and \emph{support-rank},
which we discuss in \cref{sec:intro-rank}.

\paragraph{Techniques.}
Any lower bound argument for $\PL$ needs to rule out using further number-theoretic
tricks to speed up the protocol. The lower bound boils down to the following
(informal) statement. Consider the $N \times N$ grid, $N=2^n$. Let $(\bm x, \bm y) \sim [N]^2$ be a uniformly random
point and $\bm \ell$ a line with random slope~$\bm p$ through~$(\bm x,
\bm y)$ ($\bm p$ will be a random small prime). Let $\bm
\ell^{\downarrow d}$ be the same line but shifted down by a small offset $d$ (which will be the product of all small numbers). We
will show that any sufficiently dense set $A\subset[N]^2$ cannot distinguish~$\bm \ell$ from $\bm \ell^{\downarrow d}$ in the following sense:\\[-1.3em]

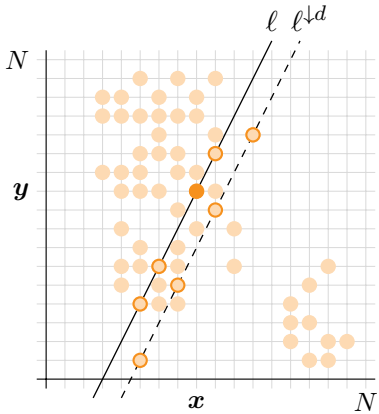
\noindent

\begin{center}
\begin{tikzpicture}[scale=0.25]
 
 

\tikzstyle{ihatetikz} = [text width=20em]
 \node[ihatetikz] at (37,10) {\emph{Line Lemma (\cref{lem:line}, informal):}\\ With high probability,
 $|A \cap \bm \ell| \approx |A \cap \bm \ell^{\downarrow d}|$.};
 

\draw[gray!30, line width=0.3pt]
  \foreach \i in {0,1,...,17} {
    (\i+0.5,0) -- (\i+0.5,18)
    (0,\i+0.5) -- (18,\i+0.5)
  };
\draw[black, line width=0.4pt]
    (0.5,0) -- (0.5,18)
    (0,0.5) -- (18,0.5);
 
\draw[black, line width=.5pt]
    (3.0,-0.5) -- (12.5,18.5);
    
\node[black, anchor=south] at (12.5,18.5) {$\ell$};

\draw[black, dashed, line width=.5pt]
    (4.5,-0.5) -- (14.0,18.5);
    
\node[black, anchor=south] at (14.5,18.5) {$\ell^{\downarrow d}$};

 
 
 
\foreach \x/\y in {3/11,3/14,3/15,
        4/5, 4/6, 4/8, 4/10, 4/11, 4/14, 4/15,
        5/4, 5/6, 5/7, 5/10, 5/11, 5/12, 5/14, 5/16,
        6/4, 6/5, 6/6, 6/10, 6/12, 6/13, 6/14, 6/15,
        7/4, 7/6, 7/7, 7/9, 7/11, 7/12, 7/14, 7/16,
        8/8, 8/10, 8/11, 8/14, 8/15,
        9/10, 9/12, 9/13, 9/16,
        10/6, 10/8,
        11/13,
        13/2, 13/3, 13/4,
        14/1, 14/3, 14/5,
        15/1, 15/2, 15/6,
        16/2
        }{
  \filldraw[color=BurntOrange!30!white,fill=BurntOrange!30!white,thick] (\x+0.5, \y+0.5) circle (10pt);
}
\foreach \x/\y in {5/1, 7/5, 9/9, 11/13}{
  \filldraw[color=BurntOrange,fill=BurntOrange!30!white,thick] (\x+0.5, \y+0.5) circle (10pt);
}

\foreach \x/\y in {5/4, 6/6, 8/10, 9/12}{
  \filldraw[color=BurntOrange,fill=BurntOrange!30!white,thick] (\x+0.5, \y+0.5) circle (10pt);
}
 
\fill[fill=BurntOrange] (8.5, 10.5) circle (11pt);
 
\small
 
\node[black] at (8+0.5, -0.7) {$\bm x$};
\node[black] at (17+0.5, -0.7) {$N$};
 
 
\node[black, anchor=east] at (0.1, 10+0.5) {$\bm y$};
\node[black, anchor=east] at (0.1, 17+0.5) {$N$};
 
\end{tikzpicture}
 
\end{center}

Our proof of this lemma relies on a new type of decomposition lemma (\cref{lem:decomposition}) that expresses any bounded function $f\colon\ZN^2\to[0,1]$ as a sum of a structured component (which is locally periodic) and a pseudorandom component (which is unbiased over lines). We discuss our proof techniques in more detail in~\cref{sec:overview}. For now, we spend the rest of this introduction motivating the study of \textsc{Point--Line Incidence} and explaining the implications of \cref{thm:cc-main}.

\begin{table}[t]
\centering
\renewcommand{\arraystretch}{1}
\newcommand{\pad}{\hskip 1.1em}
\setlength\tabcolsep{.2em}
\begin{tabular}{l @{\pad} l @{\pad}l @{\pad}l @{\pad} l @{\pad} r}
\toprule[.5mm]
\bf Problem
& \bf Task
& \bf \boldmath $\Rcc$
& \bf \boldmath $\suprank$
& \bf \boldmath $\signrank$ &
\bf Reference \\
\midrule
\textsc{Equality} & $y=b$ &  $\Theta(1)$ & 2$^\dagger$ & 3 \\
\textsc{Greater-Than} & $y\geq b$ &  $\Theta(\log n) $ & $\exp(\Theta(n))$ & 2
& \cite{Braverman2015,Viola2015,Srinivasan2023}\\
\textsc{Halfplane Membership} & $y\geq ax+b$ &  $\Theta(n) $ & $\exp(\Theta(n))$ & 3
& \cite{ACHS24} \\
\textsc{Point--Line Incidence} & $y = ax+b$ &  $\Theta(\log n) $ & $3^\dagger$ & $\Theta(1)$
& This paper (\hyperref[thm:cc-main]{Thm.~\ref{thm:cc-main}})\\
\bottomrule[.5mm]
\end{tabular}
\caption{Arithmetic problems together with their randomised communication complexity~($\Rcc$), support-rank ($\suprank$), and sign-rank ($\signrank$). Here Alice and Bob get $n$-bit integers $(x,y)$ and~$(a,b)$, respectively. When indicated with $\dagger$, the support-rank is given for the \emph{negated} problem.}
\label{tab}
\end{table}

\subsection{Consequences for Rank Measures}
\label{sec:intro-rank}

Since $y = ax+b$ is a simple polynomial equation, it has low algebraic complexity as formalised by \emph{support-} and \emph{sign-rank}. These  are defined for a two-party function $f\colon\cX\times\cY\to\{0,1\}$ as follows.
\begin{itemize}[leftmargin=1.5em]
\item \emph{Support-rank} $\suprank(f)$ is the smallest dimension~$r\in\N$ in which $f$ can be represented by linear equalities; that is, such that there exist embeddings $\phi \colon \mathcal
X \to \bR^r$ and $\psi \colon \mathcal Y \to \bR^r$ with
\[
    \forall x,y\colon \qquad f(x,y) = 0 \enspace\iff\enspace \langle \phi(x), \psi(y) \rangle = 0.
\]
In other words, the support-rank of $f$, viewed as a $|\cX|$-by-$|\cY|$ boolean matrix, is the least rank of a matrix $M\in\R^{\cX\times\cY}$ that has the same support as $f$. Previously, support-rank has been called \emph{nondeterministic rank} in quantum communication complexity~\cite{dWol03}, \emph{equality rank} in circuit complexity~\cite{HP2010}, and \emph{minimum rank} in graph theory~\cite{FH07}.

\item
\emph{Sign-rank} $\signrank(f)$ is the smallest dimension~$r\in\N$ in which $f$ can be represented by linear inequalities; that is, such that there exist embeddings $\phi \colon \mathcal
X \to \bR^r$ and $\psi \colon \mathcal Y \to \bR^r$ with
\[
    \forall x,y\colon \qquad
    \langle \phi(x), \psi(y) \rangle \cdot (-1)^{f(x,y)} >0.
\]
Sign-rank is known to be equivalent to \emph{unbounded-error randomised
communication}, where the two parties have private randomness and must succeed
with probability $> 1/2$~\cite{PS86}. Sign-rank is much more powerful than support-rank: Every problem of support-rank $r$ has sign-rank~$O(r^2)$ (see, e.g.,~\cite{GHIS25}), while there are problems over $n\coloneqq \log\max(|\cX|,|\cY|)$ input bits with sign-rank $2$ and support-rank $2^n$ (e.g., \textsc{Greater-Than}).
\end{itemize}
For \textsc{Point--Line Incidence}, the embeddings $\phi(x,y) = (x, y, 1)$ and~$\psi(a,b) = (a, -1, b)$ show that the support-rank of the negation $\neg\PL$ is at most~3. Consequently, \cref{thm:cc-main} implies the following separation of support-rank and public-coin randomised communication complexity~$\Rcc(f)$; previously, it was not known whether $\suprank(f) = O(1)$ implies $\Rcc(f) = O(1)$.
\begin{corollary} \label{cor:sep}
    There is an $f$ with $\suprank(f)\leq 3$ and $\Rcc(f)\geq \omega(1)$.
\end{corollary}
The analogous separation for sign-rank has been long known: \textsc{Greater-Than} has sign-rank~$2$ but randomised complexity~$\Omega(\log n)$~\cite{Braverman2015,Viola2015,Srinivasan2023}. For support-rank, the separation in \cref{cor:sep} is qualitatively the best possible, in the sense that all problems of
support-rank 2 reduce to \textsc{Equality} and therefore have randomised
complexity $O(1)$. On the other hand, proving a more dramatic quantitative separation for an $n$-bit problem remains open:
\begin{problem}
    Is there an $f$ with $\suprank(f)\leq O(1)$ and $\Rcc(f)\geq n^{\Omega(1)}$?
\end{problem}
The analogous separation for sign-rank was recently obtained by~\cite{HHL20,ACHS24}. They exhibit problems with sign-rank 3 that have randomised complexity $\Theta(n)$. In fact, this holds for the \textsc{Halfplane Membership} problem of deciding whether $y\geq ax+b$.

\subsection{Integer Inner Product, and Motivation from Oracle Protocols}
\label{sec:intro-iip}

In the \textsc{Integer Inner Product} ($\IIP^k_n$) problem Alice and Bob receive $n$-bit integers~$x_1,\ldots,x_k$ and $y_1,\ldots,y_k$, respectively, and they want to check whether $\sum_{i=1}^k x_i y_i = 0$. As already mentioned, this problem has randomised complexity $O(k \log
n)$ \cite{CLV19}. We get a matching lower bound:

\begin{corollary}
\label{cor:iip}
The public-coin randomised communication complexity of $\IIP^k_n$ is
$\Omega(k\log n)$ for every~$3\leq k \leq n^\epsilon$ where $\epsilon > 0$ is
some constant.
\end{corollary}

The condition $k \leq n^\epsilon$ is most likely an artefact of our technique. Nevertheless, it is only a mild restriction since for large $k$ there is also an $\Omega(k)$ lower bound by a reduction from \textsc{Disjointness}.

\textsc{Integer Inner Product} was introduced by~\cite{CLV19} to show that efficient randomised protocols cannot be simulated by the \textsc{Equality} oracle. Specifically, they showed that $\IIP^6_n$ has no efficient protocol if we are only allowed to use randomness to solve instances of \textsc{Equality}. In their words, ``Equality alone does not simulate randomness.'' In fact,
\textsc{Equality} doesn't help in this problem at all: a deterministic protocol
with access to an \textsc{Equality} oracle still requires $\Omega(n)$ bits of
communication. We write this result as $\Dcc^{\textsc{Eq}}(\IIP^6_n)
\geq \Omega(n)$, where $\Dcc^{\textsc{Eq}}(f)$ is the number of
\textsc{Equality} queries required to compute $f$. This was improved to hold for
\textsc{PL} by \cite{Cheung23} and simplified by \cite{Cheung25ITCS}. It remains
open whether it is possible to push this type of separation further: the
papers~\cite{Cheung23,GHR25} ask whether a similar lower bound against
\textsc{Equality}-oracle protocols holds for some problem with constant
randomised complexity:
\begin{problem}
Is there a problem $f$ with $\Rcc(f) \leq O(1)$ but $\Dcc^{\textsc{Eq}}(f) \geq \Omega(n)$?
\end{problem}
This was the initial context for the conjecture $\Rcc(\PL) \geq \omega(1)$ of
\cite{Cheung23}, because if \textsc{Point--Line Incidence} had a constant-cost
protocol, it would have resolved this question. In light of our
\cref{thm:cc-main}, the question remains open. Currently, the best lower bound
against \textsc{Equality}-oracle protocols for a problem with $\Rcc(f)\leq O(1)$
is $\Dcc^{\textsc{Eq}}(f)\geq \Omega(\sqrt n)$~\cite{GHR25}. 

\subsection{Motivation from Constant-Cost Communication}
\label{sec:constant-cost}

Our results fit more broadly within a recent line of work on ``constant-cost'' communication. Traditionally, communication complexity tries to classify $n$-bit communication problems into \emph{easy problems} that can be solved with~$(\log n)^{O(1)}$ bits of communication and \emph{hard problems} that require~$n^{\Omega(1)}$ bits. The papers~\cite{HHH23,HWZ22} proposed to consider the \emph{constant-cost} complexity dichotomy~$O(1)$-vs-$\omega(1)$ instead of the traditional $(\log n)^{O(1)}$-vs-$n^{\Omega(1)}$.
This provides a new ``sandbox'' where we can study old questions from a novel perspective, helping us develop new lower-bound methods, and also to discover new questions interesting in their own right.

\paragraph{Class \boldmath $\BPP_0$.}
A central mystery in this new theory asks to characterise all total problems $f$
that have constant randomised complexity, $\Rcc(f)\leq O(1)$. This class of
problems, denoted $\BPP_0$, captures the most extreme ways in which randomness
can benefit protocols. It has been heavily investigated in recent
work~\cite{HHH22,EHK22,HHPTZ22,HZ24,FHHH24,FGHH25,GHR25,Balla2025}. The class
has many equivalent definitions (constant discrepancy, point--halfspace
arrangements with constant margin, PAC learnable under pure differential
privacy~\cite{FX14}), but all of them so far are ``semantic,'' hiding a promise.
Since constant-cost protocols seem very restrictive, one may hope for a simple
``syntactic'' characterisation. Various natural candidates for a
characterisation have been disproved: the class admits no complete
problem~\cite{FHHH24}; nor certain complete hierarchies~\cite{FGHH25,GHR25}.
\cref{thm:cc-main} poses a challenge for any such characterisation, which
must somehow generalise our highly-tailored lower bound argument (at least qualitatively).

\paragraph{Class \boldmath $\UPP_0$.}
The second-most prominent constant-cost class is $\UPP_0$ that contains all total problems with a constant-cost unbounded-error protocol, that is, problems of constant sign-rank~\cite{PS86}. As discussed above, \textsc{Greater-Than} shows that $\UPP_0\not\subseteq\BPP_0$ and our \cref{cor:sep} upgrades this to a separation $\SUPP_0\not\subseteq\BPP_0$ where $\SUPP_0\subseteq\UPP_0$ is the class of problems with constant support-rank, defined in~\cite{GHIS25}. The converse separation has grown into a nagging problem:
\begin{problem} \label{prob:bpp-upp}
Show that $\BPP_0\not\subseteq\UPP_0$.
\end{problem}
This separation has been shown for \emph{partial} functions by Hatami, Hosseini, and Meng~\cite{HHM23}. Towards a separation for total functions, one might first ask to show the weaker separation $\BPP_0\not\subseteq\SUPP_0$. But this already follows by considering the \textsc{Equality} problem, which has support-rank~$2^n$ on $n$-bit inputs. A better first step is the following, asked in~\cite{GHIS25}.
\begin{problem} \label{prob:bpp-psupp}
Show that $\BPP_0\not\subseteq\Pcc_0^{\SUPP}$.
\end{problem}
Here $\Pcc_0^{\SUPP}$ is the class of problems expressible as boolean combinations of constantly many problems of constant support-rank. This dispenses with \textsc{Equality} as a separating example since its complement has constant support-rank. It is known that $\Pcc_0^{\SUPP}\subseteq\UPP_0$~\cite{GHIS25} so this is truly a necessary first step towards \cref{prob:bpp-upp}. For more discussion on constant-cost communication complexity, we recommend the excellent survey of Hatami and Hatami~\cite{HH24}.

\section{Proof Overview} \label{sec:overview}

We formally view~$\PL$ as a function where Alice is given a point $(x,y)\in
[N]^2$, $N\coloneqq 2^n$, and Bob is given line parameters $(a,b)\in \{-N^2,\ldots,N^2\}^2$ and they need to compute $\PL((x,y),(a,b))=1$ iff~$y=ax+b$. Our lower bound uses the textbook discrepancy method~\cite[\S6]{Rao2020} (and thus it will hold for quantum
protocols as well). We define a pair of input distributions $(\cD_0,\cD_1)$, where~$\cD_i$ is supported over $\PL^{-1}(i)$, such that every rectangle~$A\times B$ (where $A$, $B$ are subsets of Alice's and Bob's inputs, respectively) has small discrepancy,
\begin{equation} \label{eq:disc}
    \big| \cD_0(A\times B) - \cD_1(A\times B) \big| \leq n^{-\Omega(1)}.
\end{equation}
\cref{thm:cc-main} then follows immediately from this discrepancy bound (see~\cref{sec:cc-lb}). The distributions are defined as follows:
\begin{enumerate}[label=$-$,noitemsep]
\item Choose a uniformly random point $(\bm x, \bm y) \in [N]^2$, and a
uniformly random prime $\bm p \leq W$; we explain how to choose $W$ later.
\item Let $\bm \ell_1$ be the line with slope $\bm p$ through $(\bm x, \bm y)$.
\item Let $\bm \ell_0$ be the same as $\bm \ell_1$ but shifted down by an offset parameter $d$, which will
be the product of all ``small'' numbers.
\item Let $\cD_i$ be the distribution of $((\bm x,\bm y),\bm \ell_i) \in\PL^{-1}(i)$.
\end{enumerate}

To see why we choose the offset $d$ as stated, observe
that if there is a small number $q$ not dividing~$d$, Alice can send her coordinates
modulo $q$ and Bob is able to perfectly distinguish $\cD_0$ from $\cD_1$ by
testing his equation for the line modulo $q$.

Showing the discrepancy bound in \cref{eq:disc} boils down to analysing rectangles $A\times B$ where Alice's set $A$ has large marginal probability, that is,
\[
\Pr{(\bm x,\bm y)\in A}= |A|/N^2\geq 1/n^{0.01}.
\]
Conditioned on the event $(\bm x,\bm y)\in A$, Bob's input line in $\cD_i$ becomes
$\bm\ell_i' \coloneqq (\bm \ell_i\mid (\bm x,\bm y)\in A)$.
For every choice of $A$, the maximum discrepancy~\eqref{eq:disc} over all choices of $B$ is then characterised by the \emph{total variation distance}\footnote{For distributions $\mu_0,\mu_1\in[0,1]^\Omega$ over a set $\Omega$, we define $\TV(\mu_0,\mu_1)\coloneqq \max_{B\subset \Omega}|\mu_0(B)-\mu_1(B)|=\frac{1}{2}\|\mu_0-\mu_1\|_1$.}, $\TV$, between $\bm \ell_0'$ and $\bm \ell_1'$. Hence our goal becomes to show
\begin{equation} \label{eq:tv}
\TV(\bm \ell_0', \bm \ell_1') \leq n^{-\Omega(1)}.
\end{equation}
To this end, we note that for all lines $\ell$,
\begin{equation} \label{eq:densities}
\frac{\P[\bm \ell_1'=\ell]}{\P[\bm \ell_0'=\ell]}
= \frac{|A\cap \ell|}{|A^{\uparrow d}\cap \ell|}
= \frac{|A\cap \ell|}{|A\cap \ell^{\downarrow d}|},
\end{equation}
where for a set $A\subseteq\Z^2$ we denote by $A^{\uparrow d}\coloneqq A+(0,d)$ and $A^{\downarrow d}\coloneqq A-(0,d)$ its translations up/down by $d$. To show~\eqref{eq:tv}, we need to prove that the ratio~\eqref{eq:densities} is very close to $1$ for typical lines~$\ell$. To carry out this plan, we proceed to develop tools to understand intersection sizes as in~\eqref{eq:densities}.





\subsection{Input Grid Embedding}

It will be technically convenient for us to work in the abelian group $\Z_N^2=\Z_N\times\Z_N$. Thus, we'll actually think of the input domain of $\PL$ as corresponding to a smaller grid $[M]^2\subseteq\Z^2_N$ where~$M\leq N/2$. This input grid is then naturally embedded in the lower-left corner of $\Z_N^2$ (this is drawn in blue in the illustration below). We define a line with slope $a\in\N$ through the origin that is capped to the square~$(-M,M)^2$ (and then reduced modulo $N$) as
\begin{equation}
\label{eq:def-segment}
\ell_{a,M}\coloneqq\{(x,y)\bmod N \;\colon\; y=ax;\, x,y\in\{-M+1,\ldots,M-1\}\}\subseteq\Z_N^2.
\end{equation}
We consider these lines anchored at an input grid point $(x,y)\in[M]^2$:
\[
    \ell\coloneqq \ell_{a,M}+(x,y).
\]
Lines of this type are \emph{safe}: when restricted to the input grid, $\ell\cap [M]^2$, they precisely correspond to Bob's inputs in $\PL$. Since $M\leq N/2$, these lines do not ``wrap around'' inside $[M]^2$ despite the modulo-$N$ arithmetic. Here is an illustration of a line with slope $a=2$ anchored at $(x,y)$.

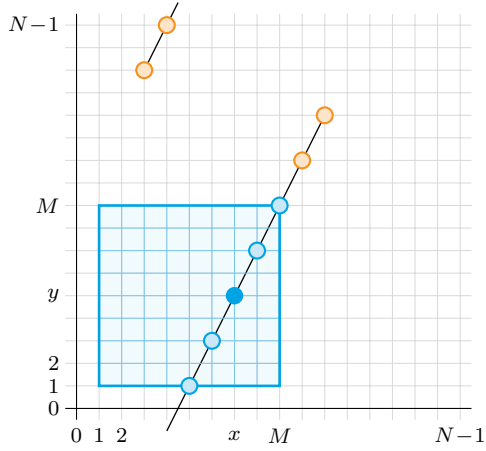
\begin{center}
    
\begin{tikzpicture}[scale=0.3]
 
 
 
\fill[ProcessBlue!5] (1.5,1.5) rectangle (9.5,9.5);
 
\draw[gray!30, line width=0.3pt]
  \foreach \i in {0,1,...,17} {
    (\i+0.5,0) -- (\i+0.5,18)
    (0,\i+0.5) -- (18,\i+0.5)
  };
\draw[black, line width=0.4pt]
    (0.5,0) -- (0.5,18)
    (0,0.5) -- (18,0.5);
 
\draw[black, line width=.5pt]
    (4.5,-0.5) -- (11.5,13.5)
    (3.5,15.5) -- (5,18.5);

\draw[Cerulean!50!white, line width=0.3pt]
  \foreach \i in {1,2,...,9} {
    (\i+0.5,1.5) -- (\i+0.5,9.5)
    (1.5,\i+0.5) -- (9.5,\i+0.5)
  };
 
 
\draw[Cerulean,line width=1pt] (1.5,1.5) rectangle (9.5,9.5);
 
\foreach \x/\y in {3/15, 4/17,10/11, 11/13}{
  \filldraw[color=BurntOrange,fill=BurntOrange!20!white,thick] (\x+0.5, \y+0.5) circle (10pt);
}
\foreach \x/\y in {5/1, 6/3, 8/7, 9/9}{
  \filldraw[color=Cerulean,fill=Cerulean!20!white,thick] (\x+0.5, \y+0.5) circle (10pt);
}
 
\fill[fill=Cerulean] (7.5, 5.5) circle (11pt);
 
\scriptsize
 
\foreach \i in {0,1,...,2}{
  \node[black] at (\i+0.5, -0.7) {\i};
}
\node[black] at (7+0.5, -0.7) {$x$};
\node[black] at (17+0.5, -0.7) {$N\!-\!1$};
\node[black] at (9+0.5, -0.7) {$M$};
 
 
\foreach \i in {0,1,...,2}{
  \node[black, anchor=east] at (0.1, \i+0.5) {\i};
}
\node[black, anchor=east] at (0.1, 5+0.5) {$y$};
\node[black, anchor=east] at (0.1, 17+0.5) {$N\!-\!1$};
\node[black, anchor=east] at (0.1, 9+0.5) {$M$};
 
\end{tikzpicture}
 
\end{center}

\subsection{Notation}
We equip functions $f\colon \Z^2_N \to\R$ with the~$L^p$-norm,
\[
\|f\|_p\coloneqq \big(\E_{z\in\Z^2_N}\big[|f(z)|^p\big]\big)^{1/p}.
\]
For $f,g\colon \Z^2_N \to\R$ we define their \emph{cross-correlation} $f\star g\colon \Z^2_N\to\R$ by
\[
(f \star g)(z) \coloneqq \E_{z'\in \Z^2_N}\big[f(z')g(z+z')\big].
\]
This is related to convolution ``$*$'' by $f\star g \coloneqq f'* g$ where $f'(z)\coloneqq f(-z)$. For a set $A\subset \Z^2_N$, we write~$\One_A\colon\Z^2_N\to\{0,1\}$ for its indicator function.
By a slight abuse of notation, we identify $A$ with its \emph{density function} $A \colon \ZN^2 \to [0,N^2]$ defined by
\[
    A(z) \coloneqq \tfrac{N^2}{|A|} \One_A(z).
\]
Under this convention, $A$ has $L^1$-norm
\[
    \|A\|_1 = \tfrac{N^2}{|A|} \E_{z \in \ZN^2}[\One_A(z)] = 1,
\]
and the expression $(A\star f)(z)$ computes the average of $f$ over the translate $A+z$:
\[
(A \star f)(z) = \E_{z'\in A}\big[f(z+z')\big].
\]
\emph{Young's inequality} says that such a smoothing operation can only decrease the $L^2$-norm:
\begin{equation}\label{eq:young}
\|A\star f\|_2 \leq \|A\|_1\cdot\|f\|_2 \leq \|f\|_2. \tag{Young}
\end{equation}

\subsection{Technical Lemmas}

To understand intersection sizes $|A\cap\ell|$ as they appear in~\cref{eq:densities}, we may now reformulate these quantities more abstractly using the notation introduced above. The density of a set $A$ relative to a line $\ell\coloneqq \ell_{a,M}+z$ is expressed as
\[
(\ell_{a,M} \star \One_A)(z)
= \E_{z'\in \ell_{a,M}}[\One_A(z+z')]
= \tfrac{1}{|\ell|}\big|A\cap \big(\ell_{a,M}+z\big)\big|.
\]
The following lemma states that these densities typically change very little when we shift a line down by~$d$. Here we use~$\cP\subseteq \N$ to denote the set of prime numbers and we set~$\cP_W\coloneqq \cP\cap[W]$.
\begin{restatable}[Line Lemma]{lemma}{linelemma}
\label{lem:line}
For every $n\in\N$ there exist $2\leq d,W\leq \exp(n^{0.4})$ such that for every $N,M\geq 2^n$ and function $f\colon \Z_N^2\to [0,1]$, we have
\begin{equation}\label{eq:line}
\E_{p\in\cP_W}\big[\big\| \ell_{p,M} \star f-\ell_{p,M}^{\downarrow d} \star f \big\|_2\big] \leq O(1/n^{0.1}).
\end{equation}
\end{restatable}

Our claimed discrepancy bound~\eqref{eq:disc} follows from this lemma; see~\cref{sec:cc-lb} for the proof.

Our proof of \nameref{lem:line} relies on a new decomposition lemma, which is our main technical contribution. It asserts that any function~$f\colon \ZN^2\to [0,1]$ can be split into a \emph{structured} component that has a local period and a \emph{pseudorandom} component that is unbiased in typical lines. Decompositions of a similar nature arise frequently in number theory, Fourier analysis, and ergodic theory; we refer the reader to \cite{Tao2007} for a survey of this perspective.

\begin{lemma}[Decomposition Lemma]
\label{lem:decomposition}
For every $\epsilon > 0$ there exist $2\leq d,W\leq \exp(1/\epsilon^4)$ such that for every $N,M\geq \exp(1/\epsilon^5)$, every function $f\colon \ZN^2\to [0,1]$ can be written as~$f=\fstr+\fpsd$ where
\begin{enumerate}[leftmargin=3.5em,label={\upshape (D\arabic*)},noitemsep]
\item \label{it:str}
$\fstr\colon \ZN^2\to \R$ satisfies
$\|\fstr-\fstr^{\uparrow d} \|_2 \leq O(\epsilon)$. 
\item \label{it:psd}
$\fpsd\colon \ZN^2\to \R$ satisfies $\E_{p\in\cP_W}\big[\|\ell_{p, M} \star f_\mathrm{psd}\|_2\big] \leq O(\epsilon)$.
\end{enumerate}

\end{lemma}
This lemma readily implies \nameref{lem:line}; see \cref{sec:line} below for the short proof. The proof of \nameref{lem:decomposition}, on the other hand, uses Fourier analysis and some number theory; see~\cref{sec:decomposition} for an overview and the proof.

\subsection{Proof of \nameref{lem:line}} \label{sec:line}

Apply \nameref{lem:decomposition} with $\epsilon\coloneqq 1/n^{0.1}$ to obtain the decomposition~$f=\fstr+\fpsd$ with associated parameters~$d,W\leq \exp(n^{0.4})$. Write $\ell_p\coloneqq\ell_{p,M}$, $\ell^{\downarrow}\coloneqq \ell^{\downarrow d}$ for short. We need to verify \cref{eq:line}. We expand it via the triangle inequality:
\begin{equation}
\text{LHS}\eqref{eq:line}
\leq \E_p\big[\| (\ell_p -\ell_p^\downarrow) \star \fstr \|_2\big]
+ \E_p\big[\| (\ell_p -\ell_p^\downarrow) \star \fpsd \|_2\big].
\end{equation}
We'll show that each of these two terms is $O(\epsilon)$, which will complete the proof. To bound the first term, we use Young's inequality, $\|\ell_p\|_1=1$, and the property \ref{it:str}:
\begin{align*}
\| (\ell_p -\ell_p^\downarrow) \star \fstr \|_2
= \| \ell_p \star (\fstr - \fstr^\uparrow)\|_2
\leq \|\ell_p\|_1\cdot \|\fstr - \smash{\fstr^\uparrow}\|_2
\leq O(\epsilon).
\end{align*}
To bound the second term, consider it for a fixed $p\in\cP_W$:
\[
\| (\ell_p -\ell_p^\downarrow) \star \fpsd \|_2
\leq \| \ell_p \star \fpsd\|_2 + \|\ell_p^\downarrow\star \fpsd\|_2
= 2 \| \ell_p \star \fpsd\|_2,
\]
where we used that $\| \ell_p \star \fpsd\|_2 = \|\smash{\ell_p^\downarrow}\star \fpsd\|_2$. 
Taking expectation over $p\in\cP_W$ and applying~\ref{it:psd} shows that the second term is $O(\epsilon)$, as desired.

\section{Proof of \nameref{lem:decomposition}}
\label{sec:decomposition}

\subsection{Fourier-Analytic Preliminaries}

Let $e(t)$ be shorthand for $e^{2\pi i t}$ where $t\in\R$.
The \define{Fourier transform} of a function $f\colon\ZN^2\to\C$ is the function $\smash{\hat{f}}\colon\ZN^2\to\C$ given by
\[
\hat{f}(\xi,\eta)\coloneqq
\E_{( x,  y)\in \ZN^2}[f( x, y) \overline{e_{\xi,\eta}(x,y)}]
\qquad\text{where}\qquad
e_{\xi,\eta}(x,y)\coloneqq e\big(\tfrac{\xi  x+\eta  y}{N}\big).
\]
The \define{Fourier inversion formula} tells us that any $f$ can be recovered from its Fourier coefficients:
\begin{equation}
\label{eqn_inversion_formula}
f~=\sum_{(\xi,\eta) \in \ZN^2} \hat{f}(\xi,\eta) e_{\xi,\eta}.
\end{equation}
We also have the following \emph{Parseval's identity}:
\begin{equation}
\|f\|_2^2 ~= \sum_{(\xi,\eta) \in \ZN^2} |\hat{f}(\xi,\eta)|^2.
\label{eq:parseval}
\end{equation}
For real-valued $f$, this gives (using $\widehat{f \star g}(\xi, \eta) = \hat{f}(-\xi, -\eta) \hat{g}(\xi, \eta)=\overline{\hat{f}(\xi, \eta)} \hat{g}(\xi, \eta)$),
\begin{equation}
\|f \star g\|_2^2 ~= \sum_{(\xi,\eta)\in\Z^2_N} |\hat{f}(\xi, \eta)|^2\cdot |\hat{g}(\xi, \eta)|^2.
\label{eqn_cross_correlation}
\end{equation}

\subsection{Proof Outline}

\paragraph{Choosing the decomposition.}
Our decomposition $f = \fstr + \fpsd$ will be obtained by partitioning the Fourier coefficients of~$f$ into two parts. The partitioning strategy is inspired by the classic \emph{circle method} in number theory (due to Hardy and Littlewood) and proceeds as follows in our setting. When should a Fourier term~$\smash{\hat{f}}(\xi, \eta) e_{\xi,\eta}$ be included in the structured part? It is when~$e_{\xi,\eta}$ is nearly unchanged under a shift by~$d$, that is, when the following difference is small:
\begin{equation}\label{eq:diff}
\Big|e_{\xi,\eta}(x,y-d)-e_{\xi,\eta}(x,y)\Big|=
\left|\frac{e_{\xi,\eta}(x,y-d)}{e_{\xi,\eta}(x,y)} -1\right|=
\left|e\bigg(\frac{\eta d}{N}\bigg)-1\right|.
\end{equation}
Note that when $\eta d/N$ is an integer, this difference is $0$. The idea is to include $e_{\xi,\eta}$ in the structured part when $\eta d/N$ is \emph{nearly} an integer so that~\eqref{eq:diff} is \emph{nearly} 0---here we are using the estimate
\begin{equation} \label{eq:torus}
\forall t\in\mathbb{R}\colon\quad
|e(t)-1| = \Theta(\|t\|_{\mathbb{T}})
\qquad\text{where}\qquad
\|t\|_{\mathbb{T}} \coloneqq \min_{n \in \mathbb{Z}} |t - n|.
\end{equation}
We will set $d$ to be the product of all ``small'' numbers, $d\coloneqq Q!$, where $Q$ is a parameter. Then~$\eta d/N$ is nearly an integer when the frequency $\eta/N\in[0,1)$ is near a rational number $a/q$ with a small denominator $q\leq Q$.
Formally, we classify $\eta\in\Z_N$ as a \emph{major arc} ($\mathfrak{M}_{T}$) if the frequency $\eta/N$ is near (distance $\leq 1/T$) a small-denominator rational; otherwise it is a \emph{minor arc} ($\mathfrak{m}_{T}$):
\begin{align*}
\mathfrak{M}_T\coloneqq\bigcup_{\substack{1\leq q\leq Q\\ \gcd(a,q)=1}} \Bigg\{\eta\in\ZN: \bigg|\frac{\eta}{N}-\frac{a}{q}\bigg|\leq \frac{1}{T} \Bigg\}
\qquad\text{and}\qquad
\mathfrak{m}_T\coloneqq \ZN\setminus\mathfrak{M}_T.
\end{align*}
Our candidate decomposition is then $\fstr\coloneqq f_{\mathfrak{M}_T}$ and $\fpsd\coloneqq f_{\mathfrak{m}_T}$, where
\begin{equation} \label{eq:decomp}
f_{\mathfrak{M}_{T}} \coloneqq\!\!
\sum_{\xi \in \ZN,\eta \in \mathfrak{M}_{T}}\!\!
\hat{f}(\xi,\eta) e_{\xi,\eta}
\qquad\text{and}\qquad
f_{\mathfrak{m}_{T}} \coloneqq\!\! 
\sum_{\xi \in \ZN, \eta \in \mathfrak{m}_{T}}\!\!
\hat{f}(\xi,\eta) e_{\xi,\eta}.
\end{equation}
\paragraph{Verifying \ref{it:str} and \ref{it:psd}.}
We now verify that \eqref{eq:decomp} satisfies \nameref{lem:decomposition} when the parameters are set as
\[
    Q \coloneqq 10\epsilon^{-3} , \qquad
    d \coloneqq Q! , \qquad
    W \coloneqq 10e^Q , \qquad
    T \coloneqq 10\epsilon^{-1} d.
\]
Consider a major arc $\eta \in \mathfrak{M}_{T}$. By definition of $\mathfrak{M}_{T}$, we have that $a/q - 1/T\leq  \eta/ N \leq a/q + 1/T$ for some $q \leq Q$. Multiplying this by $d = Q!$ shows that
\begin{equation} \label{eq:major-torus}
\left\|\frac{\eta d}{N}\right\|_{\mathbb{T}} \leq \frac{d}{T} \leq \epsilon.
\end{equation}
Write $\delta_{(x,y)}\coloneqq N^2\One_{(x,y)}$ for the point mass density at $(x,y)\in\Z^2_N$. Note that $\hat{\delta}_{(x,y)}(\xi,\eta)=\overline{e_{\xi,\eta}(x,y)}$. We can now verify \ref{it:str}:
\begin{align*}\textstyle
\|\fstr-\fstr^{\uparrow d} \|_2^2
&\textstyle = \|(\delta_{(0,0)}-\delta_{(0,d)})\star \fstr\|_2^2 \\
&\textstyle = \sum_{\xi\in \ZN,\eta\in\mathfrak{M}_T} \left|1-e\left(\tfrac{\eta d}{N}\right)\right|^2\cdot |\hat{f}(\xi, \eta)|^2
\tag{Using~\eqref{eqn_cross_correlation}}
\\
&\textstyle \leq \sum_{\xi\in \ZN,\eta\in\mathfrak{M}_T} O\Big(\left\|\tfrac{\eta d}{N}\right\|_{\mathbb{T}}\Big)^2\cdot |\hat{f}(\xi, \eta)|^2
\tag{Using~\eqref{eq:torus}}
\\
&\leq O(\epsilon^2).
\tag{Using~\eqref{eq:major-torus}, \eqref{eq:parseval}, $\|f\|_2\leq 1$}
\end{align*}
To verify \ref{it:psd}, we apply the next lemma, where plugging in our parameter
values yields the desired $O(\epsilon)$ bound. The remainder of this
section is dedicated to its proof.%
\begin{restatable}[Minor Arc Bound]{lemma}{minorarclines}
\label{lem:minor}
For $Q, W, T$ such that $e^Q<W<\sqrt{T}/2$:
\[
\E_{ p \in \cP_W}\big[ \|\ell_{ p, M} \star f_{\mathfrak{m}_{T}}\|_2^2\big]
\leq O\left(\frac{\log{\log{Q}}}{Q} + \frac{T^2W^2}{{M}^2}\right).
\]
\end{restatable}

\subsection{Proof of \nameref{lem:minor}}

Writing $\ell_p\coloneqq\ell_{p,M}$ for short, we have
\begin{align*}
\E_{ p \in \cP_W}\big[ \|\ell_p \star f_{\mathfrak{m}_T}\|_2^2\big]~
&= \sum_{\xi\in\ZN,\eta\in\mathfrak{m}_T}
\E_{ p \in \cP_W}\big[|\hat{\ell}_p(\xi,\eta)|^2\big]
\cdot |\hat{f}(\xi, \eta)|^2
\tag{Using~\eqref{eqn_cross_correlation}}\\
&\leq \max_{\xi\in\Z_N,\eta\in\mathfrak{m}_T}
\E_{ p \in \cP_W}\big[|\hat{\ell}_p(\xi,\eta)|^2\big].
\tag{Using~\eqref{eq:parseval}, $\|f\|_2\leq 1$}
\end{align*}
We'll bound this for each fixed $\xi\in\Z_N$ and $\eta\in\mathfrak{m}_T$. Let us start by computing $|\hat{\ell}_p(\xi,\eta)|^2$ for a fixed choice of slope $p\in\cP_W$. Define $X_p \coloneqq \{x \in (-M, M) : px \in (-M, M)\}$ as the set of all~$x$-coordinates of points in the line $\ell_p$ and note that $|X_p|=|\ell_p|\geq \Omega(M/W)$. Then
\begin{align}
|\hat{\ell}_p(\xi,\eta)|^2
&= \bigg|\frac{1}{|X_p|}\sum_{x \in X_p}e\big(\tfrac{\xi x + \eta p x}{N}\big)\bigg|^2
\notag \\[-0.5em]
&=\bigg|\frac{1}{|X_p|} e\big(\tfrac{\xi t +\eta pt}{N}\big)\sum_{x = 0}^{|X_p| - 1}e\big(\tfrac{\xi x + \eta px}{N}\big)\bigg|^2
\tag{where $t \coloneqq \min X_p$}
\\
&= \frac{1}{{|X_p|}^2}\left| \frac{1-e(\alpha_p|X_p|)}{1-e(\alpha_p)}\right|^2
\tag{where $\alpha_p \coloneqq \tfrac{\xi + \eta p}{N}$}
\\
&\leq O\big(\tfrac{W^2}{M^2}\big) \frac{1}{|1-e(\alpha_p)|^2} .
\label{eq:coeff-ub}
\end{align}
Here we have two cases depending on the magnitude $|1-e(\alpha_p)|=\Theta(\|\alpha_p\|_{\mathbb{T}})$. We define the set of primes for which this magnitude is small ($\ll 1/T$) and large ($\gg 1/T$) as
\[
    \mathcal{P}_{\xi,\eta,T} \coloneqq \bigg\{p\in \cP: \bigg\|\frac{\xi+ p\eta}{N}\bigg\|_{\T}\leq \frac{1}{2T}\bigg\}
    \qquad\text{and}\qquad
\overline{\mathcal{P}}_{\xi,\eta,T}=\cP\smallsetminus\cP_{\xi,\eta,T}.
\]
We now justify the following calculation that will finish the proof:
\begin{align*}
\E_{p \in \cP_W}\big[|\hat{\ell}_p(\xi,\eta)|^2\big]
&=
\E_{ p \in \cP_W}\big[
\One_{\overline{\mathcal{P}}_{\xi,\eta,T}}(p) \cdot
|\hat{\ell}_p(\xi,\eta)|^2\big]
+
\E_{ p \in \cP_W}\big[
\One_{\mathcal{P}_{\xi,\eta,T}}(p) \cdot
|\hat{\ell}_p(\xi,\eta)|^2\big]
\\
&\leq
\E_{ p \in \overline{\mathcal{P}}_{\xi,\eta,T}}\big[
|\hat{\ell}_p(\xi,\eta)|^2\big]
+
|\mathcal{P}_{\xi,\eta,T}\cap [W]|/|\cP_W|
\tag{$|\hat{\ell}_p(\xi,\eta)|\leq 1$}
\\
&\leq
O\left(\frac{T^2W^2}{{M}^2}\right)
+ O\left(\frac{\log{\log{Q}}}{Q}\right).
\end{align*}
In the last inequality, every $p\in \overline{\cP}_{\xi,\eta,T}$ satisfies $|1-e(\alpha_p)|\geq \Omega(\|\alpha_p\|_{\mathbb{T}})\geq \Omega(1/T)$ and plugging this in~\eqref{eq:coeff-ub} gives the bound on the first term. On the other hand, the following lemma (proved in the remainder of this section) gives the bound on the second term.
\begin{lemma}[Prime Bound]
\label{lem:prime-dense}
Let $\xi\in\ZN$, $\eta\in\mathfrak{m}_{T}$, and suppose $e^Q<W<\sqrt{T}/2$. Then
\begin{equation}
\label{eq:prime-mult}
\frac{|\mathcal{P}_{\xi,\eta,T}\cap [W]|}{|\cP_W|}
\leq O\left(\frac{\log{\log{Q}}}{Q}\right).
\end{equation}
\end{lemma}

\subsection{Proof of \nameref{lem:prime-dense}}

For this proof, we need the Siegel--Walfisz theorem~\cite[Corollary 11.19]{Montgomery_Vaughan_2006} that shows an upper bound on the density of primes in arithmetic progressions.
\begin{theorem}[Siegel--Walfisz]\label{thm:sw}
There exists a constant $C>0$ such that the following holds:  for any $c>0$ and any $a, q, W$ with $\gcd(a, q) = 1$ and $q \leq (\log{W})^c$ it holds that
\[
|\cP_W \cap (q\Z + a)| = \Oh\left(\frac{W}{\phi(q)\log{W}}\right) + \Oh_c\left(W \exp(-C\log{W})\right),
\]
where $\phi(\cdot)$ is Euler's totient function.
\end{theorem}

Below, we will establish the next claim:
\begin{claim}
\label{clm:ap}
For every $\xi,\eta\in\ZN$ there exist $q\in\N$, $a,r\in\{0,1,\ldots,q-1\}$ with $\gcd(a,q)=1$ s.t.
\[
\mathcal{P}_{\xi,\eta,T}\cap [W]\subset q\Z+r\qquad\text{and}\qquad\bigg|\frac{\eta}{N}- \frac{a}{q} \bigg|\leq \frac{1}{T}.
\]
\end{claim}

With this claim, we proceed to show the upper bound~\eqref{eq:prime-mult}.
If $\mathcal{P}_{\xi,\eta,T} = \emptyset$, the bound is trivial. Suppose it is non-empty. We use \cref{clm:ap} to find $a, q, r$ associated with $\xi,\eta$. Note that we must have $q \geq Q$, because otherwise $\eta\in\mathfrak{M}_T$, contradicting our choice of $\eta \in \mathfrak{m}_{T}$. Suppose first that~$q \geq (\log{W})^2$. Then $|\mathcal{P}_{\xi,\eta,T}\cap [W]| \leq W/q$ since $\mathcal{P}_{\xi,\eta,T} \subseteq q\Z + r$. The bound~\eqref{eq:prime-mult} follows from
\[
    |\mathcal{P}_{\xi,\eta,T} \cap [W]| \leq \frac{W}{q} \leq \frac{W}{(\log{W})^2} \leq \Oh\left(\frac{|\cP_W|}{\log{W}}\right) \leq \Oh\left(\frac{|\cP_W|}{Q}\right).
\]
Otherwise, suppose then that $Q \leq q \leq (\log{W})^2$. Here we can apply \cref{thm:sw}:
\begin{equation}\label{eq:sw-apply}
    |\mathcal{P}_{\xi,\eta,T}\cap [W]| \leq O\left(\frac{W}{\phi(q) \log{W}}\right) .
\end{equation}
It is known (e.g.~\cite[Thm~2.9]{Montgomery_Vaughan_2006}) that $\phi(q) = \Omega(q/\log\log
q)$. This gives $\phi(q)\geq \Omega(Q / \log\log Q)$ since $q \geq Q$. Plugging this in~\eqref{eq:sw-apply} yields the desired bound~\eqref{eq:prime-mult}, completing the proof.

\begin{proof}[Proof of \cref{clm:ap}]
If $\mathcal{P}_{\xi,\eta,T}\cap [W]$ is the empty set or contains exactly one element then we can simply take $q= \frac{N}{\gcd(N,\eta)}$ and $a= \frac{\eta}{\gcd(N,\eta)}$ and the first part of the claim follows readily. So for the remainder of this proof, let us assume that $\mathcal{P}_{\xi,\eta,T}\cap [W]$ contains at least two elements.

It follows from the definition of $\mathcal{P}_{\xi,\eta,T}$ and the triangle inequality that for all $p_1,p_2\in \mathcal{P}_{\xi,\eta,T}\cap [W]$ with $p_1<p_2$ we have 
\[
\bigg\|\frac{p_2\eta}{N}- \frac{p_1\eta}{N}\bigg\|_{\T}\leq \frac{1}{T}.
\]
This means there exists some integer $a(p_1,p_2)$ such that
\[
\bigg|\frac{(p_2-p_1)\eta}{N} -a(p_1,p_2) \bigg|\leq \frac{1}{T},
\]
which implies
\[
\bigg|\frac{\eta}{N}- \frac{a(p_1,p_2)}{(p_2-p_1)} \bigg|\leq \frac{1}{T}.
\]
Using the triangle inequality again, we conclude that for all $p_1,p_2,p_3,p_4\in \mathcal{P}_{\xi,\eta,T}\cap [W]$ with $p_1<p_2$ and $p_3<p_4$ we have
\[
\bigg|\frac{a(p_1,p_2)}{(p_2-p_1)}- \frac{a(p_3,p_4)}{(p_4-p_3)} \bigg|\leq \frac{2}{T}.
\]
Since $T/2 > W^2$, whereas $(p_2-p_1)$ and $(p_4-p_3)$ are no larger than $W$, we must have
\[
\frac{a(p_1,p_2)}{(p_2-p_1)}=\frac{a(p_3,p_4)}{(p_4-p_3)}.
\]
In other words, there exists a fraction $a/q$ with $\gcd(a,q)=1$ such that for all $p_1,p_2\in \mathcal{P}_{\xi,\eta,T}\cap[W]$ with $p_1<p_2$,
\[
\frac{a(p_1,p_2)}{(p_2-p_1)} = \frac{a}{q}.
\]
This implies that for all $p_1,p_2\in \mathcal{P}_{\xi,\eta,T}\cap [W]$ the difference $p_2-p_1$ is divisible by $q$ and hence for some $r\in \{0,1,\ldots,q-1\}$ we have $\mathcal{P}_{\xi,\eta,T}\cap[W]\subset q\Z+r$, which yields the claim.
\end{proof}

\section{Communication Lower Bounds} \label{sec:cc-lb}

Our communication lower bounds use the discrepancy method. The
\emph{discrepancy} of $f\colon\cX\times\cY\to\{0,1\}$ relative to a distribution
$\cD$ over $\cX\times\cY$ is defined by
\[
    \disc(f,\cD) \coloneqq
    \max_R | \cD_0(R) - \cD_1(R) |
\]
where the maximum is over all rectangles $R=A\times B$ with $A\subseteq\cX$, $B\subseteq\cY$,
and
\[
    \cD_b(R) = \mathbb{P}_{\bm z \sim \cD}[ \bm z \in R \wedge f(\bm z) = b ] .
\]
The \emph{discrepancy bound} for a function $f$ is
\[
\Disc(f) \coloneqq \max_\cD \log(1/\disc(f,\cD)),
\]
where the maximum is over all distributions over $\cX\times\cY$. We have the following basic fact.
\begin{fact}[{\cite[\S6]{Rao2020}}]
\label{fact:disc}
$\Rcc(f)\geq\Omega(\Disc(f))$ for all $f$.
\end{fact}
\subsection{Point--Line Incidence}

We use the \nameref{lem:line} to prove a discrepancy bound.
We restate the lemma here for convenience.

\linelemma*

\cref{thm:cc-main} follows from the next discrepancy bound, combined with
\cref{fact:disc}.
\begin{lemma} \label{lem:disc}
$\Disc(\PL) = \Omega(\log n)$.
\end{lemma}
\begin{proof}
Let $M \coloneqq 2^n$ so that Alice's inputs are $(x,y) \in [M]^2$. Let $N =
4M$. Let $d, W \leq \exp(n^{0.4})$ be obtained from the \nameref{lem:line}. We
define distributions $\cD_0$ and $\cD_1$ following the instructions of
\cref{sec:overview}, and define $\cD \coloneqq (\cD_0 + \cD_1)/2$.

Let $A \times B$ be any rectangle (so $A \subseteq [M]^2$ is a set of points,
and $B$ a set of lines).
Let $\bm \ell_1$ be a random line obtained by choosing $(\bm x, \bm y) \sim A$
uniformly at random, choosing $\bm p \sim \cP_W$ uniformly at random, and taking
the line with slope $\bm p$ through $(\bm x, \bm y)$. Let
$\bm \ell_0$ be independently distributed in a similar way, but taking
the line through $(\bm x, \bm y - d)$ instead. We may upper bound the discrepancy
on $A \times B$ using the TV distance:
\begin{equation}
\label{eq:disc-to-tv}
\begin{aligned}
    |\cD_0(A \times B) - \cD_1(A \times B)|
    &= \frac{|A|}{M^2} \left| \Pr{ \bm \ell_1 \in B }
        - \Pr{ \bm \ell_0 \in B } \right| \\
    &\leq \frac{|A|}{M^2} \max_C \left| \Pr{ \bm \ell_1 \in C }
        - \Pr{ \bm \ell_0 \in C } \right| \\
    &= \frac{|A|}{M^2} \cdot \TV(\bm \ell_1, \bm \ell_0) ,
\end{aligned}
\end{equation}
where the maximum is over all sets $C$ of lines. To establish the TV distance
bound, define the function $f \colon \ZN^2 \to [0,1]$ as $f(x,y) =
\One_{A}[x,y]$. Recall from \cref{eq:def-segment} that
\[
    \ell_{p,M} \coloneqq \{ (h, ph) \in [-M+1,M-1]^2 \colon h \in \Z \} .
\]
For a fixed line $\ell \subset [M]^2$ with slope $p \leq W$, we extend it into
$\ell^* \subset \ZN^2$ as
\[
    \ell^* \coloneqq (u,v) + \ell_{p,M} ,
\]
where $(u,v) \in \ell$ is chosen arbitrarily.
Observe that
$\ell = \ell^* \cap [M]^2$ since $\ell \subseteq (u,v) + \ell_{p,M}$ for all $(u,v) \in \ell$,
and $\ell^* \subseteq [-M, 2M]^2$, so there are no
``wraparounds'' in $\ZN^2$. For every $(x,y) \in \ell^*$,
\begin{align*}
    \ell_{p,M} \star f(x,y)
    = \E_{(h,ph) \in \ell_{p,M}} [ f(x+h, y+ph) ]
    = \frac{1}{|\ell_{p,M}|} |A \cap \ell | ,
\end{align*}
since $\ell \subseteq (x,y) + \ell_{p,M}$ by definition. Similarly,
\[
    \ell_{p,M}^\downarrow \star f(x,y)
    = \ell_{p,M} \star f(x,y + d)
    = \frac{1}{|\ell_{p,M}|} |A \cap \ell^\downarrow| .
\]
Observe that $|\ell^*| = |\ell_{p,M}|$. 
Therefore, we may write
\begin{align*}
    \Pr{ \bm \ell_1 = \ell}
        &= \Pr{ \bm p = p } \cdot \frac{|A \cap \ell|}{|A|} 
        = \Pr{ \bm p = p } \frac{1}{|A|}
            \cdot \sum_{(x,y) \in \ell^*} \ell_{p,M} \star f(x,y) ,\\ 
            \intertext{and, similarly,}
    \Pr{ \bm \ell_0 = \ell} &= \Pr{ \bm p = p } \cdot \frac{|A \cap \ell^\downarrow|}{|A|}
        = \Pr{ \bm p = p } \frac{1}{|A|}
            \cdot \sum_{(x,y) \in \ell^*} \ell_{p,M}^\downarrow \star f(x,y) .
\end{align*}
Then we can use the \nameref{lem:line} to bound the TV distance by
\begin{align*}
    2 \cdot \TV(\bm \ell_1, \bm \ell_0)
    &= \sum_{p \in \cP_W} \sum_{\text{lines } \ell \text{ of slope } p}
        \left| \Pr{ \bm \ell_1 = \ell } - \Pr{ \bm \ell_0 = \ell } \right| \\
    &= \frac{1}{|A|} \cdot \E_{\bm p}\left[ \sum_{\ell \text{ of slope } \bm p }
        \left|\sum_{(x,y) \in \ell^*} \ell_{\bm p, M} \star f(x,y)
            - \ell_{\bm p, M}^\downarrow \star f(x,y) \right| \right] \\
\intertext{For each $p$, the line segments $\ell^*$ are disjoint, so:}
    &\leq \frac{1}{|A|} \cdot \E_{\bm p}\left[
        \sum_{(x,y) \in [-M,2M]^2} \left| \ell_{\bm p, M} \star f(x,y)
            - \ell_{\bm p, M}^\downarrow \star f(x,y) \right| \right] \\
    &\leq \frac{3M}{|A|} \cdot \E_{\bm p}\left[
        \left(\sum_{(x,y) \in [-M,2M]^2} \left( \ell_{\bm p, M} \star f(x,y)
            - \ell_{\bm p, M}^\downarrow \star f(x,y) \right)^2 \right)^{1/2} \right]
            \tag{Cauchy--Schwarz}\\
    &\leq \frac{3M}{|A|} \cdot N \cdot \E_{\bm p}\left[
        \| \ell_{\bm p, M} \star f - \ell_{\bm p, M}^\downarrow \star f \|_2  \right] \\
    &\leq O\left(\frac{M N}{|A| n^{0.1}}\right) 
        \tag{\nameref{lem:line}} .
\end{align*}
Using \cref{eq:disc-to-tv}, this gives a bound of
\[
    |\cD_0(A \times B) - \cD_1(A \times B)|
    \leq O\left( \frac{N}{M n^{0.1}} \right)
    = O(1/n^{0.1}) . \qedhere
\]
\end{proof}

\subsection{Integer Inner Product} \label{sec:iip}

We now prove~\cref{cor:iip}. Recall that $\PL$ is a special case of $\IIP_n^3\colon\cZ^3\times\cZ^3\to\{0,1\}$, where $\cZ$ is the set of $n$-bit integers, in the sense that $\PL$ is a submatrix of $\IIP_n^3$. Consider the function~$\And_k \circ \IIP_n^3$ that first evaluates $k$ copies of $\IIP_n^3$ and then outputs their logical-AND, that is,
\[
(\And_k \circ \IIP_n^3)(x,y) \coloneqq
\And_k(\IIP_n^3(x^1,y^1),\ldots,\IIP_n^3(x^k,y^k)),
\]
where $x\coloneqq(x^1,\ldots,x^k)$ and $x^i\coloneqq (x^i_1,x^i_2,x^i_3)\in\cZ^3$ and similarly for $y$. We claim that $\And_k \circ \IIP_n^3$ reduces to~$\IIP_n^{3k}$ via a randomised reduction, which will show that
\begin{equation}
\label{eq:and}    
\Rcc(\And_k \circ \IIP_n^3) \leq O(\Rcc(\IIP_n^{3k})).
\end{equation}
Indeed, suppose $(x,y)$ are the inputs to $\And_k \circ \IIP_n^3$. We let Alice replace her input $x$ by
\[
\bm z\odot x\coloneqq (\bm z_1 x^1,\ldots ,\bm z_k x^k),
\]
where Alice chooses $\bm z \in \{-1,1\}^k$ uniformly at random. Then:
\begin{itemize}[label=$-$,noitemsep]
\item If $\langle x^i,y^i\rangle =0$ for all $i\in[k]$, then $\langle \bm z \odot x,y\rangle =0$ with probability 1.
\item If $\langle x^i,y^i\rangle \neq 0$ for some $i\in[k]$, then $\langle \bm z \odot x,y\rangle  \neq 0$ with probability $\geq 1/2$.
\end{itemize}
These two properties show that any randomised protocol for $\IIP_n^{3k}$ can be used to derive a randomised protocol for $\And_k\circ \IIP_n^3$, proving~\cref{eq:and}.

It remains to show that for every $k\leq n^\epsilon$ where $\epsilon>0$ is a sufficiently small constant,
\begin{equation} \label{eq:comp}
\Rcc(\And_k\circ\PL)\geq \Omega(k\log n).
\end{equation}
To this end, we employ the following $\And$-composition lemma from~\cite[Lemma~10]{GJPW2018}. The lemma there is originally stated with a measure \textsc{2WAPP} (aka ``smooth rectangle bound'') in place of $\Disc$, but the latter is a lower bound on the former~\cite{Jain2010}.
\begin{lemma}[{\cite{GJPW2018}}] \label{lem:and}
$\Disc(f)\leq O(\Rcc(\And_k\circ f)/k + \log\Rcc(\And_k\circ f))$ for all $f$.
\end{lemma}
Instantiating this with $f\coloneqq \PL$ gets us
\begin{align*}
\log n
&\leq O(\Disc(\PL)) \tag{\cref{lem:disc}} \\
&\leq O(\Rcc(\And_k\circ \PL)/k + \log\Rcc(\And_k\circ \PL)) \tag{\cref{lem:and}}\\
&\leq O(\Rcc(\And_k\circ \PL)/k + \log (k\log n)) \\
&\leq O(\Rcc(\And_k\circ \PL)/k + \epsilon\log n). \tag{$k\leq n^\epsilon$}
\end{align*}
Choosing $\epsilon>0$ small enough and rearranging gives~\cref{eq:comp}, as desired.

\bigskip\medskip
\subsubsection*{Acknowledgements}
We thank Hamed Hatami, Kaave Hosseini, Shachar Lovett, and Raghu Meka for discussions.
Special thanks to Oliver G\"o\"os for serving as an uncritical sounding board during the writing of the paper.
M.G.~and A.S.\ are supported by the Swiss State Secretariat for Education, Research, and Innovation (SERI) under contract number
MB22.00026. F.K.R.~was supported by the Swiss National Science Foundation grant TMSGI2-211214.

\pagebreak

\DeclareUrlCommand{\Doi}{\urlstyle{sf}}
\renewcommand{\path}[1]{\small\Doi{#1}}
\renewcommand{\url}[1]{\href{#1}{\small\Doi{#1}}}
\bibliographystyle{alphaurl}
\bibliography{references}

\newcommand{\etalchar}[1]{$^{#1}$}
 \providecommand{\FOCS}{Proceedings of the Symposium on Foundations of Computer Science (FOCS)} \providecommand{\SODA}{Proceedings of the Symposium on Discrete Algorithms (SODA)} \providecommand{\STOC}{Proceedings of the Symposium on Theory of Computing (STOC)} \providecommand{\CCC}{Proceedings of the Conference on Computational Complexity (CCC)} \providecommand{\ITCS}{Proceedings of the Innovations in Theoretical Computer Science Conference (ITCS)} \providecommand{\ICALP}{Proceedings of the International Colloquium on Automata, Languages, and Programming (ICALP)} \providecommand{\ICML}{Proceedings of the International Conference on Machine Learning (ICML)} \providecommand{\COLT}{Proceedings of the Conference on Learning Theory (COLT)} \providecommand{\AISTATS}{Proceedings of the International Conference on Artificial Intelligence and Statistics (AISTATS)} \providecommand{\TOCT}{ACM Transactions on Computation Theory (TOCT)} \providecommand{\RANDOM}{International Conference on Randomization and Computation
  (RANDOM)} \providecommand{\SOCG}{Proceedings of the Symposium on Computational Geometry (SoCG)} \providecommand{\FSTTCS}{Proceedings of the Conference on Foundations of Software Technology and Theoretical Computer Science (FSTTCS)} \providecommand{\JACM}{Journal of the ACM} \providecommand{\SIAMJOC}{SIAM Journal on Computing} \providecommand{\TOC}{Theory of Computing} \providecommand{\TOIT}{IEEE Transactions on Information Theory} \providecommand{\NEURIPS}{Advances in Neural Information Processing Systems (NeurIPS)} \providecommand{\JMLR}{Journal of Machine Learning Research} \providecommand{\SOSA}{Symposium on Simplicity in Algorithms (SOSA)}
\begin{thebibliography}{GJPW18}

\bibitem[ACHS24]{ACHS24}
Manasseh Ahmed, Tsun-Ming Cheung, Hamed Hatami, and Kusha Sareen.
\newblock Communication complexity and discrepancy of halfplanes.
\newblock In {\em \SOCG}, pages 5:1--5:17. Schloss Dagstuhl, 2024.
\newblock \href {https://doi.org/10.4230/LIPIcs.SoCG.2024.5} {\path{doi:10.4230/LIPIcs.SoCG.2024.5}}.

\bibitem[BHT25]{Balla2025}
Igor Balla, Lianna Hambardzumyan, and István Tomon.
\newblock Factorization norms and an inverse theorem for {M}ax{C}ut.
\newblock In {\em \FOCS}, pages 947--963. IEEE, 2025.
\newblock \href {https://doi.org/10.1109/focs63196.2025.00049} {\path{doi:10.1109/focs63196.2025.00049}}.

\bibitem[BW15]{Braverman2015}
Mark Braverman and Omri Weinstein.
\newblock A discrepancy lower bound for information complexity.
\newblock {\em Algorithmica}, 76(3):846--864, 2015.
\newblock \href {https://doi.org/10.1007/s00453-015-0093-8} {\path{doi:10.1007/s00453-015-0093-8}}.

\bibitem[CHH{\etalchar{+}}25]{Cheung25ITCS}
Tsun-Ming Cheung, Hamed Hatami, Kaave Hosseini, Aleksandar Nikolov, Toniann Pitassi, and Morgan Shirley.
\newblock A lower bound on the trace norm of boolean matrices and its applications.
\newblock In {\em \ITCS}, volume 325 of {\em LIPIcs}, pages 37:1--37:15. Schloss Dagstuhl, 2025.
\newblock \href {https://doi.org/10.4230/LIPIcs.ITCS.2025.37} {\path{doi:10.4230/LIPIcs.ITCS.2025.37}}.

\bibitem[CHHS23]{Cheung23}
Tsun-Ming Cheung, Hamed Hatami, Kaave Hosseini, and Morgan Shirley.
\newblock Separation of the factorization norm and randomized communication complexity.
\newblock In {\em \CCC}, volume 264 of {\em LIPIcs}, pages 1:1--1:16. Schloss Dagstuhl, 2023.
\newblock \href {https://doi.org/10.4230/LIPIcs.CCC.2023.1} {\path{doi:10.4230/LIPIcs.CCC.2023.1}}.

\bibitem[CLV19]{CLV19}
Arkadev Chattopadhyay, Shachar Lovett, and Marc Vinyals.
\newblock Equality alone does not simulate randomness.
\newblock In {\em \CCC}, pages 14:1--14:11. Schloss Dagstuhl, 2019.
\newblock \href {https://doi.org/10.4230/LIPIcs.CCC.2019.14} {\path{doi:10.4230/LIPIcs.CCC.2019.14}}.

\bibitem[dW03]{dWol03}
Ronald de~Wolf.
\newblock Nondeterministic quantum query and communication complexities.
\newblock {\em \SIAMJOC}, 32(3):681--699, 2003.
\newblock \href {https://doi.org/10.1137/s0097539702407345} {\path{doi:10.1137/s0097539702407345}}.

\bibitem[EHK22]{EHK22}
Louis Esperet, Nathaniel Harms, and Andrey Kupavskii.
\newblock Sketching distances in monotone graph classes.
\newblock In {\em \RANDOM}, pages 18--1. Schloss Dagstuhl, 2022.
\newblock \href {https://doi.org/10.4230/LIPIcs.APPROX/RANDOM.2022.18} {\path{doi:10.4230/LIPIcs.APPROX/RANDOM.2022.18}}.

\bibitem[FGHH25]{FGHH25}
Yuting Fang, Mika G\"o\"os, Nathaniel Harms, and Pooya Hatami.
\newblock Constant-cost communication does not reduce to $k$-{H}amming distance.
\newblock In {\em \STOC}, 2025.
\newblock \href {https://doi.org/10.48550/arXiv.2407.20204} {\path{doi:10.48550/arXiv.2407.20204}}.

\bibitem[FH07]{FH07}
Shaun Fallat and Leslie Hogben.
\newblock The minimum rank of symmetric matrices described by a graph: A survey.
\newblock {\em Linear Algebra and its Applications}, 426(2–3):558--582, 2007.
\newblock \href {https://doi.org/10.1016/j.laa.2007.05.036} {\path{doi:10.1016/j.laa.2007.05.036}}.

\bibitem[FHHH24]{FHHH24}
Yuting Fang, Lianna Hambardzumyan, Nathaniel Harms, and Pooya Hatami.
\newblock No complete problem for constant-cost randomized communication.
\newblock In {\em \STOC}, 2024.
\newblock \href {https://doi.org/10.48550/arXiv.2404.00812} {\path{doi:10.48550/arXiv.2404.00812}}.

\bibitem[FX14]{FX14}
Vitaly Feldman and David Xiao.
\newblock Sample complexity bounds on differentially private learning via communication complexity.
\newblock In {\em \COLT}, pages 1000--1019. PMLR, 2014.
\newblock \href {https://doi.org/10.48550/arXiv.1402.6278} {\path{doi:10.48550/arXiv.1402.6278}}.

\bibitem[GHIS25]{GHIS25}
Mika G{\"o}{\"o}s, Nathaniel Harms, Valentin Imbach, and Dmitry Sokolov.
\newblock Sign-rank of $k$-{H}amming distance is constant.
\newblock In {\em \FOCS}, pages 2353--2368. IEEE, December 2025.
\newblock \href {https://doi.org/10.1109/focs63196.2025.00123} {\path{doi:10.1109/focs63196.2025.00123}}.

\bibitem[GHR25]{GHR25}
Mika G\"{o}\"{o}s, Nathaniel Harms, and Artur Riazanov.
\newblock Equality is far weaker than constant-cost communication.
\newblock In {\em \RANDOM}, volume 353 of {\em LIPIcs}, pages 58:1--58:14. Schloss Dagstuhl, 2025.
\newblock \href {https://doi.org/10.4230/LIPIcs.APPROX/RANDOM.2025.58} {\path{doi:10.4230/LIPIcs.APPROX/RANDOM.2025.58}}.

\bibitem[GJPW18]{GJPW2018}
Mika G{\"o}{\"o}s, T.S. Jayram, Toniann Pitassi, and Thomas Watson.
\newblock Randomized communication vs. partition number.
\newblock {\em ACM Transactions on Computation Theory}, 10(1):4:1--4:20, 2018.
\newblock \href {https://doi.org/10.1145/3170711} {\path{doi:10.1145/3170711}}.

\bibitem[HH24]{HH24}
Hamed Hatami and Pooya Hatami.
\newblock Guest column: Structure in communication complexity and constant-cost complexity classes.
\newblock {\em ACM SIGACT News}, 55(1):67--93, 2024.
\newblock \href {https://doi.org/10.1145/3654780.3654788} {\path{doi:10.1145/3654780.3654788}}.

\bibitem[HHH22]{HHH22}
Lianna Hambardzumyan, Hamed Hatami, and Pooya Hatami.
\newblock A counter-example to the probabilistic universal graph conjecture via randomized communication complexity.
\newblock {\em Discrete Applied Mathematics}, 322:117--122, 2022.
\newblock \href {https://doi.org/10.1016/j.dam.2022.07.023} {\path{doi:10.1016/j.dam.2022.07.023}}.

\bibitem[HHH23]{HHH23}
Lianna Hambardzumyan, Hamed Hatami, and Pooya Hatami.
\newblock Dimension-free bounds and structural results in communication complexity.
\newblock {\em Israel Journal of Mathematics}, 253(2):555--616, 2023.
\newblock \href {https://doi.org/10.1007/s11856-022-2365-8} {\path{doi:10.1007/s11856-022-2365-8}}.

\bibitem[HHL20]{HHL20}
Hamed Hatami, Kaave Hosseini, and Shachar Lovett.
\newblock Sign rank vs discrepancy.
\newblock In {\em \CCC}, pages 18--1. Schloss Dagstuhl, 2020.
\newblock \href {https://doi.org/10.4230/LIPIcs.CCC.2020.18} {\path{doi:10.4230/LIPIcs.CCC.2020.18}}.

\bibitem[HHM23]{HHM23}
Hamed Hatami, Kaave Hosseini, and Xiang Meng.
\newblock A {Borsuk-Ulam} lower bound for sign-rank and its applications.
\newblock In {\em \STOC}, pages 463--471, 2023.
\newblock \href {https://doi.org/10.1145/3564246.3585210} {\path{doi:10.1145/3564246.3585210}}.

\bibitem[HHP{\etalchar{+}}22]{HHPTZ22}
Hamed Hatami, Pooya Hatami, William Pires, Ran Tao, and Rosie Zhao.
\newblock Lower bound methods for sign-rank and their limitations.
\newblock In {\em \RANDOM}, pages 22--1. Schloss Dagstuhl, 2022.
\newblock \href {https://doi.org/10.4230/LIPIcs.APPROX/RANDOM.2022.22} {\path{doi:10.4230/LIPIcs.APPROX/RANDOM.2022.22}}.

\bibitem[HP10]{HP2010}
Kristoffer~Arnsfelt Hansen and Vladimir Podolskii.
\newblock Exact threshold circuits.
\newblock In {\em \CCC}, pages 270--279. IEEE, 2010.
\newblock \href {https://doi.org/10.1109/ccc.2010.33} {\path{doi:10.1109/ccc.2010.33}}.

\bibitem[HWZ22]{HWZ22}
Nathaniel Harms, Sebastian Wild, and Viktor Zamaraev.
\newblock Randomized communication and implicit graph representations.
\newblock In {\em \STOC}, 2022.
\newblock \href {https://doi.org/10.1145/3519935.3519978} {\path{doi:10.1145/3519935.3519978}}.

\bibitem[HZ24]{HZ24}
Nathaniel Harms and Viktor Zamaraev.
\newblock Randomized communication and implicit representations for matrices and graphs of small sign-rank.
\newblock In {\em \SODA}, pages 1810--1833. SIAM, 2024.
\newblock \href {https://doi.org/10.1137/1.9781611977912.72} {\path{doi:10.1137/1.9781611977912.72}}.

\bibitem[JK10]{Jain2010}
Rahul Jain and Hartmut Klauck.
\newblock The partition bound for classical communication complexity and query complexity.
\newblock In {\em \CCC}, pages 247--258. IEEE, 2010.
\newblock \href {https://doi.org/10.1109/CCC.2010.31} {\path{doi:10.1109/CCC.2010.31}}.

\bibitem[MV06]{Montgomery_Vaughan_2006}
Hugh Montgomery and Robert Vaughan.
\newblock {\em Multiplicative Number Theory I: Classical Theory}.
\newblock Cambridge Studies in Advanced Mathematics. Cambridge University Press, 2006.
\newblock \href {https://doi.org/10.1017/CBO9780511618314} {\path{doi:10.1017/CBO9780511618314}}.

\bibitem[PS86]{PS86}
Ramamohan Paturi and Janos Simon.
\newblock Probabilistic communication complexity.
\newblock {\em Journal of Computer and System Sciences}, 33(1):106--123, 1986.
\newblock \href {https://doi.org/10.1016/0022-0000(86)90046-2} {\path{doi:10.1016/0022-0000(86)90046-2}}.

\bibitem[RY20]{Rao2020}
Anup Rao and Amir Yehudayoff.
\newblock {\em Communication Complexity: and Applications}.
\newblock Cambridge University Press, January 2020.
\newblock \href {https://doi.org/10.1017/9781108671644} {\path{doi:10.1017/9781108671644}}.

\bibitem[SY23]{Srinivasan2023}
Srikanth Srinivasan and Amir Yehudayoff.
\newblock The discrepancy of greater-than.
\newblock Technical report, arXiv, 2023.
\newblock \href {https://doi.org/10.48550/ARXIV.2309.08703} {\path{doi:10.48550/ARXIV.2309.08703}}.

\bibitem[Tao07]{Tao2007}
Terence Tao.
\newblock Structure and randomness in combinatorics.
\newblock In {\em \FOCS}, pages 3--15. IEEE, 2007.
\newblock \href {https://doi.org/10.1109/focs.2007.17} {\path{doi:10.1109/focs.2007.17}}.

\bibitem[Vio15]{Viola2015}
Emanuele Viola.
\newblock The communication complexity of addition.
\newblock {\em Combinatorica}, 35(6):703--747, 2015.
\newblock \href {https://doi.org/10.1007/s00493-014-3078-3} {\path{doi:10.1007/s00493-014-3078-3}}.

\end{thebibliography}

\end{document}